\documentclass[11pt,fleqn]{article}
\usepackage{fullpage}
\usepackage{amsthm}

\usepackage{amsmath,amssymb}
\usepackage[boxed]{algorithm}
\usepackage[noend]{algorithmic}
\usepackage{bbm}
\usepackage{graphicx}
\usepackage{subfigure}
\usepackage{xspace}
\usepackage{color}
\usepackage{afterpage}
\usepackage{ifpdf}
\ifpdf    % we are running pdflatex
\usepackage{hyperref}
\else    % we are running latex
\usepackage[hypertex]{hyperref}
\fi
\newtheorem{theorem}{Theorem}[section]
\newtheorem{lemma}[theorem]{Lemma}

\newtheorem{corollary}[theorem]{Corollary}
\newtheorem{definition}[theorem]{Definition}

\newtheorem{notation}[theorem]{Notation}

\newcommand {\ignore} [1] {}
\DeclareMathOperator{\supp}{supp}

\newcommand{\etal}{{\em et al.\ }\xspace}

\providecommand{\eqdef}{:=}

\DeclareMathOperator{\texp}{Texp}

\title{Metric Decompositions of Path-Separable Graphs%
\thanks{The final publication of the paper is available at Springer via http://dx.doi.org/10.1007/s00453-016-0213-0.}}
\author{Lior Kamma%
\thanks{Work supported in part by a US-Israel BSF grant \#2010418
and an Israel Science Foundation grant \#897/13.
Email: \texttt{\{lior.kamma,robert.krauthgamer\}@weizmann.ac.il}
}
\\ The Weizmann Institute
\and Robert Krauthgamer\footnotemark[\value{footnote}]
\\ The Weizmann Institute
}

\begin{document}

\maketitle

\begin{abstract}
A prominent tool in many problems involving metric spaces is 
a notion of randomized low-diameter decomposition.
Loosely speaking, $\beta$-decomposition refers to a probability distribution 
over partitions of the metric into sets of low diameter, 
such that nearby points (parameterized by $\beta>0$)
are likely to be ``clustered'' together.
Applying this notion to the shortest-path metric in edge-weighted graphs,
it is known that $n$-vertex graphs admit 
an $O(\ln n)$-padded decomposition \cite{Bartal96},
and that excluded-minor graphs admit $O(1)$-padded decomposition 
\cite{KPR93,FT03,AGGNT14}.

We design decompositions to the family of $p$-path-separable graphs,
which was defined by Abraham and Gavoille~\cite{AG06}
and refers to graphs that admit vertex-separators consisting 
of at most $p$ shortest paths in the graph.

Our main result is that every $p$-path-separable $n$-vertex graph 
admits an $O(\ln (p \ln n))$-decomposition,
which refines the $O(\ln n)$ bound for general graphs, 
and provides new bounds for families like bounded-treewidth graphs.
Technically, our clustering process differs from previous ones 
by working in (the shortest-path metric of) carefully chosen subgraphs.
\end{abstract}

\section{Introduction} \label{sec:intro}
In recent decades, the problem of decomposing a metric space into low-diameter subspaces has become a key step in the solution for many problems, including metric embeddings (e.g. \cite{Assouad83,Bartal96,Rao99,GKL03,FRT04,KLMN05}), distance oracles and routing schemes design (e.g. \cite{AP90,DSB97,Talwar04,CGMZ05,MN07}), graph sparsification (e.g. \cite{EGKRTT2010,KKN14}) and optimization problems, such as multicommodity cuts \cite{KPR93,GVY96,LR99}, $0$-extension \cite{CKR04} and the traveling salesman problem \cite{Talwar04}.

Following the recent literature, we focus on \emph{randomized} decompositions, which loosely speaking refers to a probability distribution over partitions 
of a metric space into sets (called clusters) of low diameter, 
such that nearby points are more likely to be ``clustered'' together.
The formal definitions follows.

Let $(X,d)$ be a metric space,
and denote the ball of radius $\varrho>0$ around $x\in X$ by $B(x,\varrho) \eqdef \{y \in X : d(x,y) \le \varrho\}$. 
Let $\Pi$ be a partition of $X$. Every $S \in \Pi$ is called a {\em cluster},
and for every $x\in X$, let $\Pi(x)$ denote the unique cluster 
$S \in \Pi$ such that $x \in S$.
We will be using the following definition of Abraham \etal \cite{AGGNT14}.

\begin{definition}\label{def:beta} 
A metric space $(X,d)$ is called {\em $\beta$-decomposable} for $\beta >0$ if for every $\Delta>0$ there is a probability distribution $\mu$ over partitions of $X$, satisfying the following properties.
\begin{enumerate}
\renewcommand{\theenumi}{(\alph{enumi})}
	\item \label{it:DiameterBound}
	Diameter Bound: For every $\Pi \in \supp(\mu)$ and $S \in \Pi$, $diam(S) \le \Delta$.
	\item \label{it:PadProb}
	Padding: For every $x \in X$ and $0 \le \gamma \le 1/100$, $$\Pr_{\Pi \sim \mu}[B(x,\gamma \Delta) \subseteq \Pi(x) \ ] \ge 2^{- \beta \gamma} \; .$$
\end{enumerate}
\end{definition}

A slightly different definition that is common in the literature 
under the name {\em padded decomposition}, see e.g.\ \cite{KLMN05,LN2004}, 
is the special case of setting in~\ref{it:PadProb} $\gamma=1/\beta$,
i.e., requiring that for every $x \in X$, with probability at least $1/2$ 
the entire ball $B(x,\Delta/\beta)$ is contained in a single cluster of $\Pi$.
Our results provide constructions that satisfy Definition~\ref{def:beta}, 
and thus immediately apply also to the more common definition.

The metric spaces we study arise as shortest-path metrics in (certain) graphs. 
Specifically, given $G = (V,E,w)$ 
a connected graph with non-negative edge weights, let $d_G$ denote the shortest-path metric induced on $V$ by $G$. Denote by $B_G(u,\varrho)$ the ball of radius $\varrho>0$ around $u \in V$ in the metric space $(V,d_G)$.
We say that a graph $G$ is $\beta$-decomposable if the metric space $(V,d_G)$ is $\beta$-decomposable.

Bartal~\cite{Bartal96} proved that for every $n$-point metric space, $\beta = O(\log n)$, and that this bound is tight for general metric spaces, thus motivating an extensive research on restricted families of metric spaces. 
Notable progress has been made for families defined by topological restrictions, such as shortest-path metrics in graphs excluding a fixed minor \cite{KPR93,FT03,AGGNT14} or bounded-genus graphs \cite{LS10,AGGNT14} and geometric restrictions, such as a bounded doubling dimension \cite{GKL03} or hyperbolic structure \cite{KL06}.

In this paper we consider metrics induced by graphs of bounded ``path separability'', which is a blend of topological and geometric restrictions,
as defined below.

\paragraph{Shortest-Path Separators.} 
Given a graph $G = (V,E,w)$ and a set $S \subseteq V$, an {\em $S$-flap} is a connected component of $G[V \setminus S]$. We say that $S$ is a (balanced) {\em vertex separator} if every $S$-flap $U$ has size $|U| \le |V|/2$. Vertex separators are widely used in divide-and-conquer algorithms. %In many applications the efficiency of the solution is dependent on the size of the separator.
Thorup \cite{T04} observed that every planar graph has a vertex separator composed of three shortest paths \footnote{If we relax the balance constraint, and require that every $S$-flap $U$ has size $|U| \le 2|V|/3$, then two shortest paths suffice.}, and used this property to design distance and reachability oracles for planar graphs.
Abraham and Gavoille \cite{AG06} extended this notion and defined {\em path separability}. Intuitively, a graph is path separable if it has a vertex separator composed of a few shortest-paths.
\begin{notation}
For sets $X_1, \ldots, X_m$ and $S \subseteq [m]$, we denote $X_S \eqdef \bigcup_{j \in S}X_j$. 
\end{notation}

\begin{definition}\label{def:pathDec}\cite{AG06}
A graph $G=(V,E,w)$ is called {\em $p$-path separable} for $p \in \mathbb{N}$ if there exists $S \subset V$ 
such that the following holds.
\begin{enumerate}
	\item There exist $P_1,\ldots,P_m \subseteq V$ such that $S = P_{[m]}$ and every $P_j$ is the union of $p_j$ shortest-paths in $G_j \eqdef G \setminus P_{[j-1]}$.
	\item $\sum_{j \in [m]}p_j \le p$.
	\item $S$ is a vertex separator, and every $S$-flap is $p$-path-separable.
\end{enumerate}
\end{definition}
Abraham and Gavoille showed that for every graph $H$ there is a number $p = p(H)$ such that every graph $(V,E)$ that excludes $H$ as a minor is $p$-path-separable under every edge weights $w$.
Diot and Gavoille \cite{DG10} proved that every graph of treewidth $t$ is $\left\lceil (t-1)/2 \right\rceil$-path-separable with every edge weights $w$.

\subsection{Main Results}
\begin{theorem} \label{th:sepToDec}
Every $p$-path-separable graph $(V,E,w)$ is $O(\ln(p \ln |V|))$-decomposable.
\end{theorem}
We further note that if, in addition, for every subgraph $G'$ of $G$, we can find a $p$-path-separator in polynomial time then we can efficiently sample from the distribution guaranteed in Theorem~\ref{th:sepToDec}.

Combining our theorem with the result of Diot and Gavoille \cite{DG10} we get an upper bound for bounded treewidth graphs.
\begin{corollary}\label{cor:twToDec}
Every graph $(V,E)$ of treewidth $t$ with every edge weights $w$ is $O(\ln(t \ln |V|))$-decomposable.
\end{corollary}

Previously, no bound was known for $p$-path-separable graphs other than $O(\log |V|)$ due to Bartal \cite{Bartal96}.
For graphs of treewidth $t$ the known upper bound is $\beta = O(t)$ due to \cite{AGGNT14}. 
Our decomposition provides a tradeoff between $t$ and $|V|$ 
and matches or improves all other bounds when $t \ge \ln\ln|V|$. 
It is conjectured that $\beta = O(\log t)$, 
which would be tight due to the lower bound of Bartal~\cite{Bartal96},
and our result provide partial evidence in favor of this conjecture.

Many known results ``interface'' the metric only through decompositions, 
and thus plugging in our decomposition bounds immediately yields new results for the aforementioned families of metric spaces.
For example, using a result from \cite{KLMN05} we conclude that 
every $n$-vertex graph of treewidth $t$ (with every edge weights $w$) 
can be embedded in a Hilbert space with distortion $O(\sqrt{\ln(t \ln n) \cdot \ln n})$, which improves over the known bound $O(\sqrt{t \ln n})$ whenever $t \ge \ln \ln n$.

\subsection{Techniques}
\paragraph{Carving Random Balls.} A common approach for constructing a decomposition of a metric $(V,d)$ is to choose a sequence of {\em centers} $c_1,\ldots,c_k \in V$ and corresponding radii $R_1,\ldots,R_k \le \Delta/2$, where the choice of centers and/or radii may involve randomization, and then define
$$S_j = \{v \in V : j = \min\{i \in [k] : d(v,c_i) \le R_i \}\} \;.$$
Clearly, each $S_j$ has diameter at most $\Delta$.
This approach goes back to \cite{Bartal96}, and has seen many useful variations, for example, randomly ordering the centers \cite{CKR04} or reducing the number of centers \cite{CKR04,GKL03}.
As it turns out, it is enough to bound the number of centers {\em locally}. More formally, for every $v \in V$ we  control the number of centers $c_j$ that threaten $v$ in the sense that $\Pr[d(v,c_j) < R_j + \gamma \Delta] > 0$.

\paragraph{Carving Balls in Subgraphs.} Inspired by ideas from \cite{AGGNT14}, we introduce in this paper another modification to the approach described theretofore. In addition to the above, for every center $c_j$ we choose a corresponding subgraph $G_j$ of $G$, such that $c_j \in V(G_j)$, and define 
$$S_j = \{v \in V : j = \min\{i \in [k] : v \in V(G_i) \;\; and \;\; d_{G_i}(v,c_i) \le R_i \}\} \;.$$
Choosing the subgraphs and centers carefully allows us to reduce the number of centers that threaten a vertex $v$ in two ways. The first and more obvious manner is by making sure that $v \in V(G_i)$ for only a few indices $i$. The second aspect is a bit more subtle. Since distances are considered in subgraphs of $G$, they might be larger than the corresponding distances in the original graph $G$, as demonstrated in Figure~\ref{f:sub}, thus reducing the number of threateners of a vertex $v$.

\begin{figure*}[t]
  \begin{center}
    \subfigure[{Original unit weighted graph $G$. 
      When all outer-cycle vertices are centers and all $R_i \in [1,2]$,
      vertex $x$ is threatened by $8$ centers. 
    }] 
{\label{f:sub-a}\includegraphics[scale=0.4]{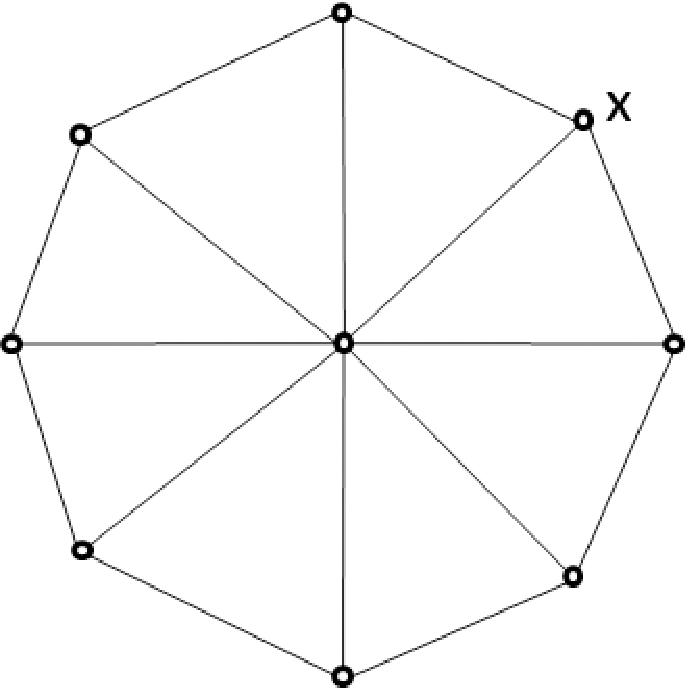}} \hspace{3pc}
    \ 
    \subfigure[If balls are carved in this subgraph,
      $x$ is threatened by only $5$ centers.]
    {\label{f:sub-b}\includegraphics[scale=0.4]{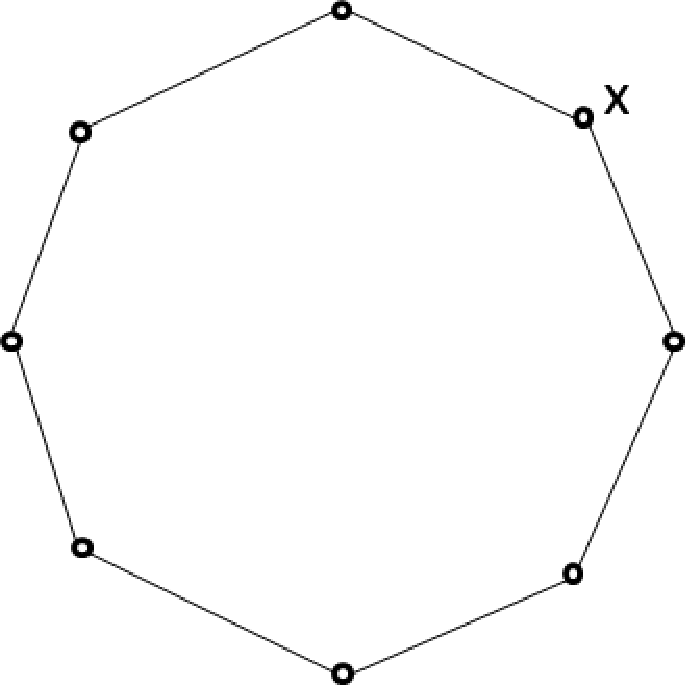}}
  \end{center}
  \caption{When carving balls in a subgraph of $G$, the number of threateners to a vertex is smaller.}
  \label{f:sub}
\end{figure*}

Note that we need to ensure that $\{S_j\}_{j \in [k]}$ is indeed a partition of $V$, i.e. that the balls $\{B_{G_j}(c_j,R_j)\}_{j \in [k]}$ cover all of $V$.

\subsection{Preliminaries}
\paragraph{The Truncated Exponential Distribution.} Define the {\em truncated exponential} with parameters $\lambda >0$ and $0 \le \alpha < \beta < \infty$, 
denoted $\texp_{[\alpha,\beta]}(\lambda)$,
to be distribution given by the probability density function 
$$g_{\lambda,[\alpha,\beta]}(x) = \frac{1}{\lambda(e^{- \alpha / \lambda} - e^{- \beta / \lambda})}e^{- x/ \lambda}  \quad \quad \forall x \in [\alpha,\beta] \;.$$
Note that this is the pdf of an exponential random variable with mean $\lambda$ that is conditioned to be in the range $[\alpha,\beta]$.
\paragraph{$r$-Nets.} Given a metric space $(V,d)$, an {\em $r$-net} of $(V,d)$ is a set $Y \subseteq V$ satisfying
\begin{enumerate}
	\item Packing: For all distinct $u,v \in Y$ we have $d(u,v)>r$.
	\item Covering: For every $v \in V$, there is some $u \in Y$ such that $d(v,u) \le r$.
\end{enumerate}
\section{Decomposing Path-Separable Graphs}
In this section we prove Theorem~\ref{th:sepToDec}. 
We present a procedure which, given a $p$-path-separable graph $G$ and a parameter $\Delta>0$, produces a random partition $\Pi$ of $V$. 
The algorithm works in two phases. The first phase, presented in detail as Algorithm~\ref{alg:pPathCen}, constructs a sequence of centers.
This is performed deterministically and by recursion. 
The algorithm finds a path-separator $S$ of $G$ and chooses a $\Delta/4$-net on each path of the separator to serve as centers. For every center vertex, the algorithm chooses a corresponding subgraph of $G$. 
The algorithm is then invoked recursively on every $S$-flap.
The second phase, presented in detail in Algorithm~\ref{alg:pPathRec}, samples random radii and carves balls around the centers to obtain a partition of $V$. The radii are all sampled independently at random from a truncated exponential distribution.

\afterpage{
\begin{algorithm}[ht]
\begin{algorithmic}[1]
\REQUIRE A $p$-path-separable graph $G$ and a parameter $\Delta$.
\ENSURE A sequence $(c_1,G_1), (c_2,G_2), \ldots$.
\STATE let $P_1,\ldots,P_m$, $G_1,\ldots,G_m$ be as in Definition~\ref{def:pathDec}.
\FOR{$j=1$ to $m$}
\FORALL{paths $P \in P_j$}
\STATE let $(c_1,c_2,\ldots)$ be a $\Delta/4$-net of $(P,d_G|_P)$ in an arbitrary order.
\STATE let $N_j$ be the sequence $(c_1,G_j), (c_2,G_j), \ldots$.
\ENDFOR
\ENDFOR
\STATE let $N$ be the concatenation of $N_1, N_2, \ldots , N_m$ in that order.
\FORALL{connected components $G'$ of $G \setminus P_{[m]}$} \label{l:rec}
\STATE invoke Choose-Centers$(G',\Delta)$ and append the output sequence to $N$.
\ENDFOR
\RETURN $N$.
\caption{Choose-Centers$(G,\Delta)$}
\label{alg:pPathCen}
\end{algorithmic}
\end{algorithm}

\begin{algorithm}[ht]
\begin{algorithmic}[1]
\REQUIRE A $p$-path-separable graph $G$ and a parameter $\Delta$.
\ENSURE A partition $\Pi$ of $V$.
\STATE let $(c_1,G_1), (c_2,G_2), \ldots$ be the sequence returned by Choose-Centers$(G,\Delta)$. 
\STATE let $\lambda \leftarrow \frac{\Delta}{10\ln (9p \log n)}$.
\STATE let $\Pi \leftarrow \emptyset$.
\FORALL{$j \ge 1$}
\STATE choose $R_j \sim \texp_{[\Delta/4,2\Delta/5]}(\lambda)$ independently at random.
\STATE let $S_j \leftarrow B_{G_j}(c_j,R_j) \setminus \left( \bigcup_{S \in \Pi}S \right)$.
\STATE let $\Pi \leftarrow \Pi \cup \{S_j\}$.
\ENDFOR
\RETURN $\Pi \setminus \{ \emptyset \}$.
\caption{Decomposing Path-Separable Graphs}
\label{alg:pPathRec}
\end{algorithmic}
\end{algorithm}
}
\vspace{1pc}
Denote by $\tilde{N}$ the union of all $\Delta/4$-nets throughout the execution of Algorithm~\ref{alg:pPathCen}. Note that $\tilde{N}$ is independent of the random radii (in fact, it is constructed deterministically). For sake of clarity, for a center $t \in \tilde{N}$ let $G_t,R_t,B_t$ be the subgraph, radius and ball corresponding to $t$ respectively.
To prove that Algorithm~\ref{alg:pPathRec} produces a partition of $V$, consider $v \in V$. During the execution of Algorithm~\ref{alg:pPathCen} there exists a subgraph $G'$ of $G$ such that Algorithm~\ref{alg:pPathCen} is (recursively) invoked on $G'$ and $v$ is in the path-separator $S$ of $G'$ chosen by the algorithm (in fact, $G'$ is unique). 
Let $P$ be the path in $S$ such that $v \in P$, let $c \in P$ be the closest net point to $v$, and let $G''$ be the respective subgraph, then $P \subseteq G''$. By the definition of a net, $d_P(v,c) \le \Delta/4$, and therefore $v \in B_P(c, \Delta/4) \subseteq B_{G''}(c, R_c)$ (recall $R_c \ge \Delta/4$ is the radius chosen for $c$). Therefore $\bigcup_{S \in \Pi}S = V$, and Algorithm~\ref{alg:pPathRec} indeed outputs a partition of $V$.

To prove the diameter requirement, let $S \in \Pi$, and let $x,y \in S$. Then there exists $t \in \tilde{N}$ such that $S \subseteq B_t \subseteq B_G(t,2\Delta/5)$ and therefore $d_G(x,y) < \Delta$.

Next, we prove the padding property of the decomposition. 
Let $x \in V$, and let $0 \le \gamma \le 1/80$. Denote $B = B_G(x,\gamma \Delta)$.
We say that $B$ is {\em settled by $t \in \tilde{N}$} if $B_t$ is the first ball (in order of execution) to have non-empty intersection with $B$. Therefore, $B \not\subseteq \Pi(x)$ iff $B$ is settled by $t$ and $B_t \cap B \ne B$ for some $t \in \tilde{N}$.
Let $\tilde{N}_x \eqdef \{ t \in \tilde{N} : \Pr[B_t \cap B \ne \emptyset] > 0 \}$ be the set of centers that threaten $x$. In order to bound the size of $\tilde{N}_x$ we consider the execution of Algorithm~\ref{alg:pPathCen}.
Consider first a single recursion level. Denote the current graph by $G'$, and let $P_1',\ldots,P_m'$ and $G_1',\ldots,G_m'$ be as in Definition~\ref{def:pathDec}. Let $j \in [m]$ and let $P$ be some path in $P_j'$.  Consider the $\Delta/4$-net $N$ picked by the algorithm. Since $P$ is a shortest-path in $G_j'$, 
$$\big|\{ t \in N : \Pr[ B_t \cap B \ne \emptyset] > 0 \}\big| \le \big|\{t \in N : B_{G_j'}(t,2\Delta/5) \cap B \ne \emptyset \}\big| \;.$$
Let $s,t \in N$ be such that $B_{G_j'}(s,2\Delta/5) \cap B \ne \emptyset$ and $B_{G_j'}(t,2\Delta/5) \cap B \ne \emptyset$. Since $P$ is a shortest path in $G_j'$, we get that $d_{P}(s,t) \le 2\Delta/5 + 2 \gamma \Delta + 2\Delta/5 < \Delta$. Therefore 
$$\big|\{t \in N : B_{G_j'}(t,2\Delta/5) \cap B \ne \emptyset \}\big| \le \frac{\Delta}{\Delta/4} = 4 \;.$$
Since in every recursive call of Algorithm~\ref{alg:pPathCen}, the number of vertices in the input graph is reduced by at least a factor of $1/2$, the depth of the recursion is at most $\log n$.
Every recursion level contains exactly one subgraph that contains $x$. Since the number of paths in each such subgraph is at most $p$, we conclude that $|\tilde{N}_x| \le 4 p \log n$. For simplicity, let us further assume this inequality holds with equality, and denote $k \eqdef 4 p \log n$.

Let $t_1,\ldots,t_k$ be the elements of $\tilde{N}_x$ in the order in which the algorithm considers them.
Denote by ${\cal E}$ the event that $B \not\subseteq \Pi(x)$, and for every $i \in [k]$, denote by ${\cal E}_i$ the event that $B$ was not settled before $t_i$ was considered.
\begin{lemma} \label{l:ind}
$\Pr[{\cal E} \mid {\cal E}_j] \le \left( 1 + \frac{k-j+1}{k^{3/2} - 1} \right)(1-e^{- 20 \gamma \ln k})$ for all $j \in [k]$.
\end{lemma}

\begin{proof}
Consider some $j \in [k]$. Conditioned on ${\cal E}_j$, there are three possible outcomes to the $j$th round. Either $B$ was not settled also in the $j$th round, and then (for $j \le k-1$) the final status of $B$ is left to the following rounds, or $B$ was settled in the $j$th round, and in that case, either $B \not\subseteq \Pi(x)$, (and thus ${\cal E}$ occurred), or $B \subseteq \Pi(x)$.
The "bad" event ${\cal E}$ can therefore only occur (either in the $j$th round or some time in the future) if either $j \le k-1$ and $R_{t_j} < d_G(x,t_j)- \gamma \Delta$ or if $\big| R_{t_j} - d_G(x,t_j) \big| \le \gamma \Delta$.

Denote $a = \max\{\Delta/4, d(x,t_j) - \gamma \Delta\}$, $b = \min\{d(x,t_j) + \gamma \Delta ,2\Delta/5 \}$, then 
\begin{equation*}
\begin{split}
\Pr\big[\; \big| R_{t_j} - d_G(x,t_j) \big| \le \gamma \Delta \; \big] &= \int_{a}^{b}{\frac{1}{\lambda(k^{-10/4} - k^{-20/5})}e^{- x/ \lambda}dx}\\
&= \frac{e^{- \frac{\Delta/4}{\lambda}}}{k^{-10/4} - k^{-20/5}} (1-e^{- \frac{b-a}{\lambda} }) \cdot e^{- \frac{a- \Delta/4}{\lambda}} \; .
\end{split}
\end{equation*}
Since $b-a \le 2 \gamma \Delta$, then $1-e^{- \frac{b-a}{\lambda}} \le 1-e^{- 20 \gamma \ln k}$. Substituting $e^{- \frac{\Delta/4}{\lambda}} = k^{-2.5}$ we get that 
\begin{equation}
\Pr\big[\; \big| R_{t_j} - d_G(x,t_j) \big| \le \gamma \Delta\; \big] \le \frac{k^{3/2}}{k^{3/2}-1} (1-e^{- 20 \gamma \ln k}) \cdot e^{- \frac{a- \Delta/4}{\lambda}}\;.
\label{eq:cutBall}
\end{equation}
Similarly we get that 
\begin{equation}
\begin{split}
\Pr[R_{t_j} < d_G(x,t_j)- \gamma \Delta] \le \frac{k^{3/2}}{k^{3/2} - 1} ( 1 - e^{- \frac{a- \Delta/4}{\lambda}}) \;.
\end{split}
\label{eq:notSettle}
\end{equation}
The proof proceeds by induction over $j= k, k-1,\ldots,1$.
As previously noted, $$\Pr[{\cal E} \mid {\cal E}_k] \le \Pr\big[\; \big| R_{t_k} - d_G(x,t_k) \big| \le \gamma \Delta\; \big] \;.$$ Plugging \eqref{eq:cutBall} we get that
\begin{equation*}
\Pr[{\cal E} \mid {\cal E}_k] \le \frac{k^{3/2}}{k^{3/2}-1} (1-e^{- 20 \gamma \ln k}) \cdot e^{- \frac{a- \Delta/4}{\lambda}} \le \left( 1 + \frac{1}{k^{3/2}-1}\right) (1-e^{- 20 \gamma \ln k}) \;,
\end{equation*}
where the last inequality follows from the fact that $a \ge \Delta/4$. 
Next, let $j \in [k-1]$. Then as previously explained,
\begin{equation}
\Pr[{\cal E} \mid {\cal E}_j] = \Pr\Big[\big| R_{t_j} - d_G(x,t_j) \big| \le \gamma \Delta\Big] +\Pr[R_{t_j} < d_G(x,t_j) - \gamma \Delta] \cdot \Pr[{\cal E} \mid {\cal E}_{j+1}] \; .
\label{eq:recursion}
\end{equation}
By plugging \eqref{eq:cutBall} and \eqref{eq:notSettle} into \eqref{eq:recursion} and using the induction hypothesis, we get that

\begin{equation*}
\begin{split}
\Pr[{\cal E} \mid {\cal E}_j] &\le
\frac{k^{3/2}}{k^{3/2} - 1}\left( e^{- \frac{a- \Delta/4}{\lambda}} + ( 1 - e^{- \frac{a - \Delta/4}{\lambda}})\left(1 + \frac{k-j}{k^{3/2}-1} \right) \right)(1-e^{- 20 \gamma \ln k})\\
&= \frac{k^{3/2}}{k^{3/2} - 1}\left( 1 + \frac{( 1 - e^{- \frac{a- \Delta/4}{\lambda}})(k-j)}{k^{3/2}-1} \right)(1-e^{- 20 \gamma \ln k}) \\
&\le \left( 1 + \frac{k-j+1}{k^{3/2} - 1} \right)(1-e^{- 20 \gamma \ln k}) \;,
\end{split}
\end{equation*}
which proves Lemma~\ref{l:ind}.
\end{proof}
To complete the proof of Theorem~\ref{th:sepToDec}, note that since $\gamma \le 1/80$ and for $k \ge 3$,
\begin{equation*}
\begin{split}
\Pr[{\cal E}] =
\Pr[{\cal E} \mid {\cal E}_1] &\le \left( 1 + \frac{k}{k^{3/2} - 1} \right)(1-e^{- 20 \gamma \ln k}) \\ 
&\le  \left( 1 + e^{-20 \gamma \ln k} \right)(1-e^{- 20 \gamma \ln k}) = (1-e^{- 40 \gamma \ln k})\; .
\end{split}
\end{equation*}
By setting $\beta = \frac{40 \ln k}{\ln 2} = O(\ln (p \ln n))$ we get that $\Pr[B(x,\gamma \Delta) \subseteq \Pi(x) \ ] \ge 2^{- \beta \gamma}$, thus proving the last part of Theorem~\ref{th:sepToDec}.

\paragraph{Acknowledgements.}
The authors thank Ittai Abraham, Alex Andoni, and Kunal Talwar for useful discussions during preliminary stages of this work.

\newcommand{\etalchar}[1]{$^{#1}$}

\end{document}